\documentclass[11pt]{amsart}
\usepackage{geometry}
\usepackage{graphicx}
\usepackage{amssymb, amsfonts, amsthm, dsfont,bm}
\usepackage{verbatim}

\newtheorem{thm}{Theorem}
\newtheorem{lemma}{Lemma}
\newtheorem{prop}{Proposition}

\newtheorem{coro}{Corollary}
\newtheorem{rem}{Remark}
\newtheorem{defi}{Definition}
\newtheorem{ex}{Example}
\newtheorem{hypo}{Hypothesis}

\def\rr{\mathbb{R}}
\def\nn{\mathbb{N}}
\def\zz{\mathbb{Z}}
\def\ee{\mathbb{E}}
\def\pp{\mathbb{P}}

\def\Leb{\mathrm{Leb}}
\def\mod{\mathrm{mod}}
\def\Lip{\mathrm{Lip}}
\def\dd{\mathrm{d}}

\def\mup{\mu_{n}\hspace{-1pt}\otimes\hspace{-1pt}P^{n}}
\def\Kxi{\widetilde{K}_{\overline{\xi}}}

\def\Cemp{\widehat{\mathrm{Cov}}_{n}}
\def\Cpert{\widetilde{\mathrm{Cov}}_{n}}

\def\E{\mathcal{E}_{n}}
\def\Epert{\mathcal{\widetilde{E}}_{n}}

\def\h{\widehat{h}_{n}}
\def\hpert{\widetilde{h}_{n}}

\def\C{\mathrm{Corr}}
\def\K{\widetilde{K}_{n,r}}

\title{Fluctuation bounds for chaos plus noise in dynamical systems}
\author{Cesar Maldonado}
\address{CPhT, CNRS-\'{E}cole Polytechnique, 91128 Palaiseau Cedex, France}
\address{Email address: \textit{\texttt{maldonado@cpht.polytechnique.fr}}}
\thanks{The author acknowledges J.-R. Chazottes, P. Collet and the anonymous reviewers for the careful reading of the manuscript, for their suggestions and corrections. The author thanks S. Galatolo and the DMA, Pisa, Italy, for their warm hospitality and where part of this work was done. The author is infinitely indebted to Adriana Aguilar Hervert}
\date{} 

\begin{document}
\begin{abstract}
We are interested in time series of the form $y_{n} = x_{n} + \xi_{n}$ where $\{ x_{n}\}$ is generated by a chaotic dynamical system and where $\xi_{n}$ models observational noise.
Using concentration inequalities, we derive fluctuation bounds for the auto-covariance function, the empirical measure, the kernel density estimator and the correlation dimension evaluated along $y_{0}, \ldots, y_{n}$, for all $n$.
The chaotic systems we consider include for instance the H\'{e}non attractor for Benedicks-Carleson parameters.
\end{abstract}
\maketitle

\section{Introduction}

Practically all experimental data is corrupted by noise, whence the importance of modeling dynamical systems perturbed by some kind of noise. In the literature one finds two principal models of noise. On one hand, the \textsl{dynamical noise model} in which the noise term evolves within the dynamics (see for instance \cite{Arn} and references therein). And on the other hand, the so-called \textsl{observational noise model}, in which the perturbation is supposed to be generated by the observation process (measurement). In this paper we focus on the latter model of noise.

Suppose that we are given a finite `sample' $y_{0},\ldots,y_{n-1}$ of a discrete ergodic dynamical system perturbed by observational noise. Consider a general ob\-ser\-va\-ble  $K(y_{0},\ldots,y_{n-1})$. We are interested in estimating the fluctuations of $K$ and its convergence properties as $n$ grows. Our main tool is concentration inequalities. Roughly speaking, concentration inequalities allow to systematically quantify the probability of deviation of an observable from its expect value, requiring that the observable is smooth enough. The systems for which concentration inequalities are available must have some degree of hyperbolicity. Indeed, in \cite{ChGo}, the authors prove that the class of non-uniformly hyperbolic maps mo\-de\-led by Young towers satisfy concentration inequalities. They  are either exponential or polynomial  depending on the tail of the corresponding return-times. Concentration inequalities is a recent topic in the study of fluctuations of observables in dynamical systems. The reader can consult \cite{Ch} for a panorama.

The article is organized as follows. In section 2, we give some general definitions concerning observational noise and concentration inequalities. We give some typical examples of systems perturbed by observational noise. In section 3, we  prove our main theorem, namely, concentration inequalities for observationally perturbed systems (observed systems). As a consequence, we obtain estimates on the deviation of any separately Lipschitz observable $K(y_{0},\ldots,y_{n-1})$ from its expected value. Section 4 is devoted to some applications. We derive a bound for the deviation of the estimator of the auto-covariance function in the observed system. We provide an estimate of the convergence in probability of the observed empirical measure. We study the $L^{1}$ convergence of the kernel density estimator for a observed system. We also give a result on the variance of an estimator of the correlation dimension in the observed system. The observables we consider here were studied in \cite{ChCS} and \cite{ChGo} for dynamical systems without observational noise.

\section{Generalities}

\subsection{Dynamical systems as stochastic process}

We consider a dynamical system $(X,T,\mu)$ where $(X,d)$ is a compact metric space and $\mu$ is a $T$-invariant probability measure. In practice, $X$ is a compact subset of $\rr^{n}$. 

One may interpret the orbits $(x, Tx, \ldots)$ as realizations of the stationary stochastic process defined by $X_{n}(x) = T^{n}x$. The finite-dimensional marginals of this process are the measures $\mu_{n}$ given by
\begin{equation}\label{marginal}
d\mu_{n}(x_{0},\ldots,x_{n-1}) = d\mu(x_{0})\prod_{i=1}^{n-1}\delta_{x_{i}=Tx_{i-1}}.
\end{equation}
Therefore, the stochasticity comes only from the initial condition. When the system is sufficiently mixing, one may expect that the iterate $T^{k}x$ is more or less independent of $x$ if $k$ is large enough.

\subsection{Observational noise}
The noise process is modeled as bounded random variables $\xi_{n}$ defined on a probability space $(\Omega, \mathcal{B},P)$ and assuming values in $X$. Without loss of generality, we can assume that the random variables $\xi_{n}$ are centered, i.e. have expectation equal to 0.

\noindent
In most cases, the noise is small and it is convenient to represent it by the random variables $\varepsilon\xi_{i}$ where $\varepsilon>0$ is the amplitude of the noise and $\xi_{i}$ is of order one.\newline 

We introduce the following definition.

\begin{defi}[Observed system]
For every $i\in\nn\cup\{0\}$ (or $i\in\zz$ if the map $T$ is invertible), we say that the sequence of points $\{y_{i}\}$ given by
\begin{equation*}
y_{i} := T^{i}x + \varepsilon\xi_{i},
\end{equation*}
is a trajectory of the dynamical system $(X,T,\mu)$ perturbed by the observational noise $(\xi_{n})$ with amplitude $\varepsilon>0$. Hereafter we refer to it simply as the \textsl{observed system}.
\end{defi}

Next, we make the following assumptions on the noise.

\textbf{Standing assumption on noise:}
\begin{enumerate}
\item $(\xi_{n})$ is independent of $X_{0}$ and $\lVert\xi_{n}\rVert\leq1$;
\item The random variables $\xi_{i}$ are independent.
\end{enumerate}


\begin{rem}
As we shall see, the $\xi_{i}$ need not be independent, although it is a natural assumption in practice.
\end{rem}

We notice that, under the same assumption on the noise, the authors of \cite{LN} give a consistent algorithm for recovering the unperturbed  time series from the sequence $\{y_{i}\}$. They assume that the process $(X_{n})$ is generated by a sufficiently chaotic dynamical system. The merit of Lalley and Nobel (\cite{LN}) is that a few assumptions are made (compare with Kantz-Schreiber's or Abarbanel's books \cite{KS,Aba}). In contrast, for the case of unbounded noise (e.g. Gaussian) and if the system present strongly homoclinic pairs of points, then with positive probability it is impossible to recover the initial condition of the true trajectory even observing an infinite sequence with noise (see also \cite{LN}).

\subsection{Examples}

\begin{ex}\label{ex1}
Consider Smale's solenoid map, $T_{S}:\rr^{3}\to\rr^{3}$ which maps the torus into itself:
\[
T_{S}\left(\phi,u,v\right) = \left( 2\phi \ \ \mod\ 2\pi, \beta u+\alpha\cos(\phi),\beta v+\alpha\sin(\phi) \right),
\]
where $0<\beta<1/2$ and $\beta<\alpha<1/2$. Let the random variables $\xi_{i}$ be uniformly distributed on the solid sphere of radius one. For every vector $x = (\phi,u,v)$ in the torus, the observed system is given by $y_{i} = T_{S}(x_{i}) + \varepsilon \xi_{i}$, for some  fixed $\varepsilon>0$.
\end{ex}

\begin{ex}
Take $\mathbb{S}^{1}$ (the unit circle) as state space. Let us fix an increasing sequence $a_{0}<a_{1}<\cdots<a_{k}=a_{0}$, and consider for each interval $(a_{j},a_{j+1})$ ($0\leq j\leq k-1$) a monotone map $T_{j}:(a_{j},a_{j+1}) \to \mathbb{S}^{1}$. The map $T$ on $\mathbb{S}^{1}$ is given by $T(x)=T_{j}(x)$ if $x\in(a_{j},a_{j+1})$. It is well known that when the map $T$ is uniformly expanding, it admits an absolutely continuous invariant measure $\mu$. It is unique under some mixing assumptions. Let $P$ be the uniform distribution on $\mathbb{S}^{1}$. The observed sequence is $y_{i} = T^{i}(x) +\varepsilon \xi_{i}$.
\end{ex}

\begin{figure}[!htb]
\centering
\includegraphics[width=0.6\textwidth]{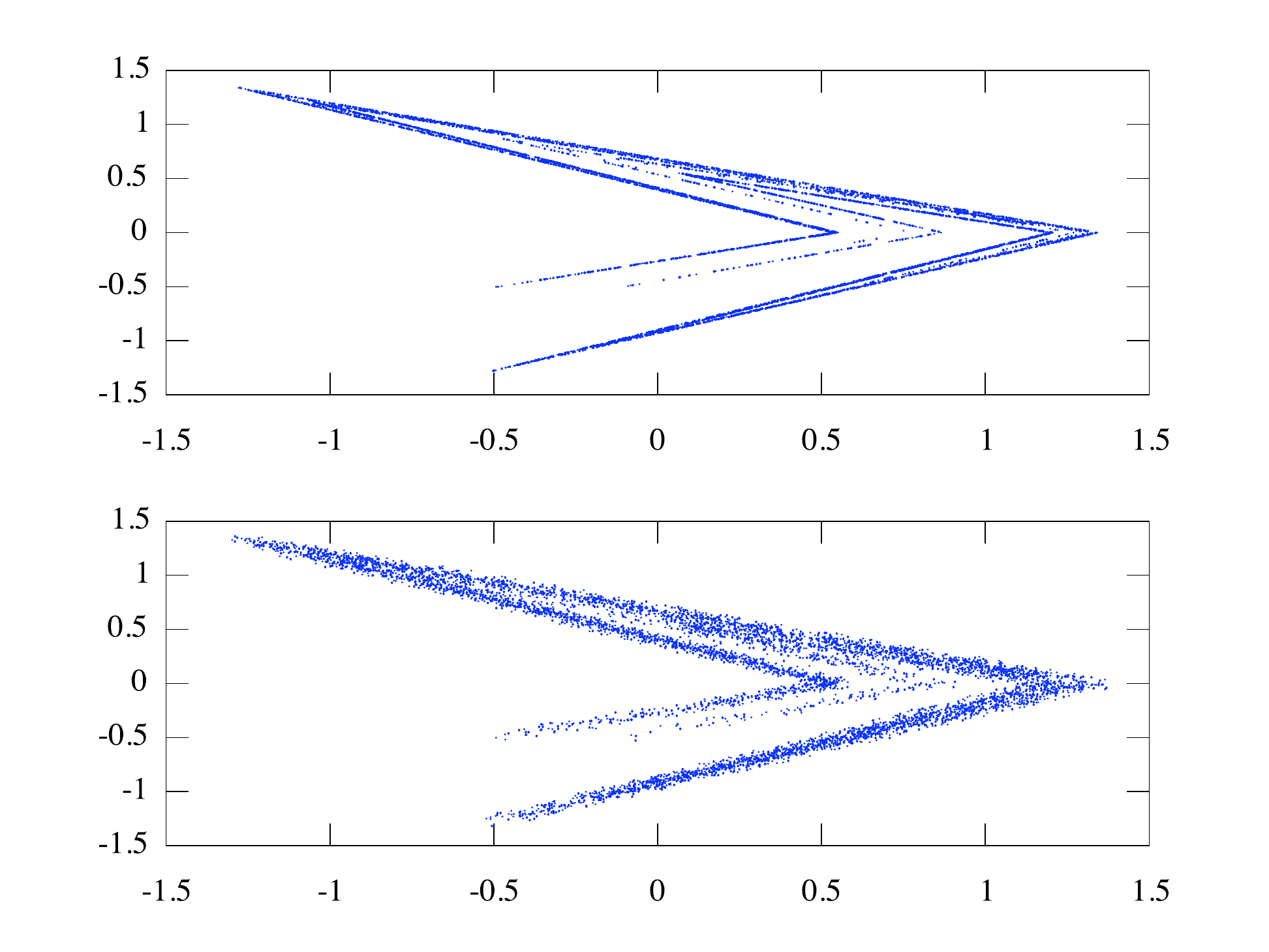}
\caption{Top: Simulation of the Lozi map for the parameters a=1.7 and b=0.5.
 Bottom: A simulation of the observed Lozi map with observational noise whose magnitude is bounded by $\varepsilon = 0.06$. }
\end{figure}
\begin{ex}
The Lozi map $T_{L}:\rr^{2}\to\rr^{2}$ is given by
\[
T_{L}(u,v) = \left( 1-a\lvert u\rvert + v ,bu \right), \hspace{1cm} (u,v)\in\rr^{2}.
\]
For $a=1.7$ and $b=0.5$ one observes numerically a strange attractor. In \cite{CL} the authors constructed a SRB measure $\mu$ for this map. It is also included in Young's framework \cite{Y1}. Now, as state space of the random variables we take $B_{1}(0)$, the ball centered at zero with radius one. Consider the uniform probability distribution on $B_{1}(0)$. Let us denote by $x$ the vector $(u,v)$ and let $\varepsilon>0$, so, the observed system is given by  $y_{i} = T_{L}^{i}x +\varepsilon \xi_{i}$. 
\end{ex}

\begin{ex}\label{ex4}
Consider the H\'enon map $T_{H}:\rr^{2}\to\rr^{2}$ defined as
\[
T_{H}(u,v) = \left( 1-au^{2} + v , bu \right), \hspace{1cm} (u,v)\in\rr^{2}.
\]
Where $0<a<2$ and $b>0$ are some real parameters. The state space of the random variables is again $B_{1}(0)$ with the uniform distribution on it. Let be $x=(u,v)$, then the observed system is given by $y_{i} = T_{H}^{i}x +\varepsilon \xi_{i}$. It is known that there exists a set of parameters $(a,b)$ of positive Lebesgue measure  for which the map $T_{H}$ has a topologically transitive attractor $\Lambda$, furthermore there exists a set $\Delta\subset\rr^{2}$ with $\Leb(\Delta)>0$ such that for all $(a,b)\in\Delta$ the map $T_{H}$ admits a unique SRB measure supported on $\Lambda$ (\cite{BY}). 
\end{ex}
\begin{figure}[!htb]
\centering
\includegraphics[width=0.7\textwidth]{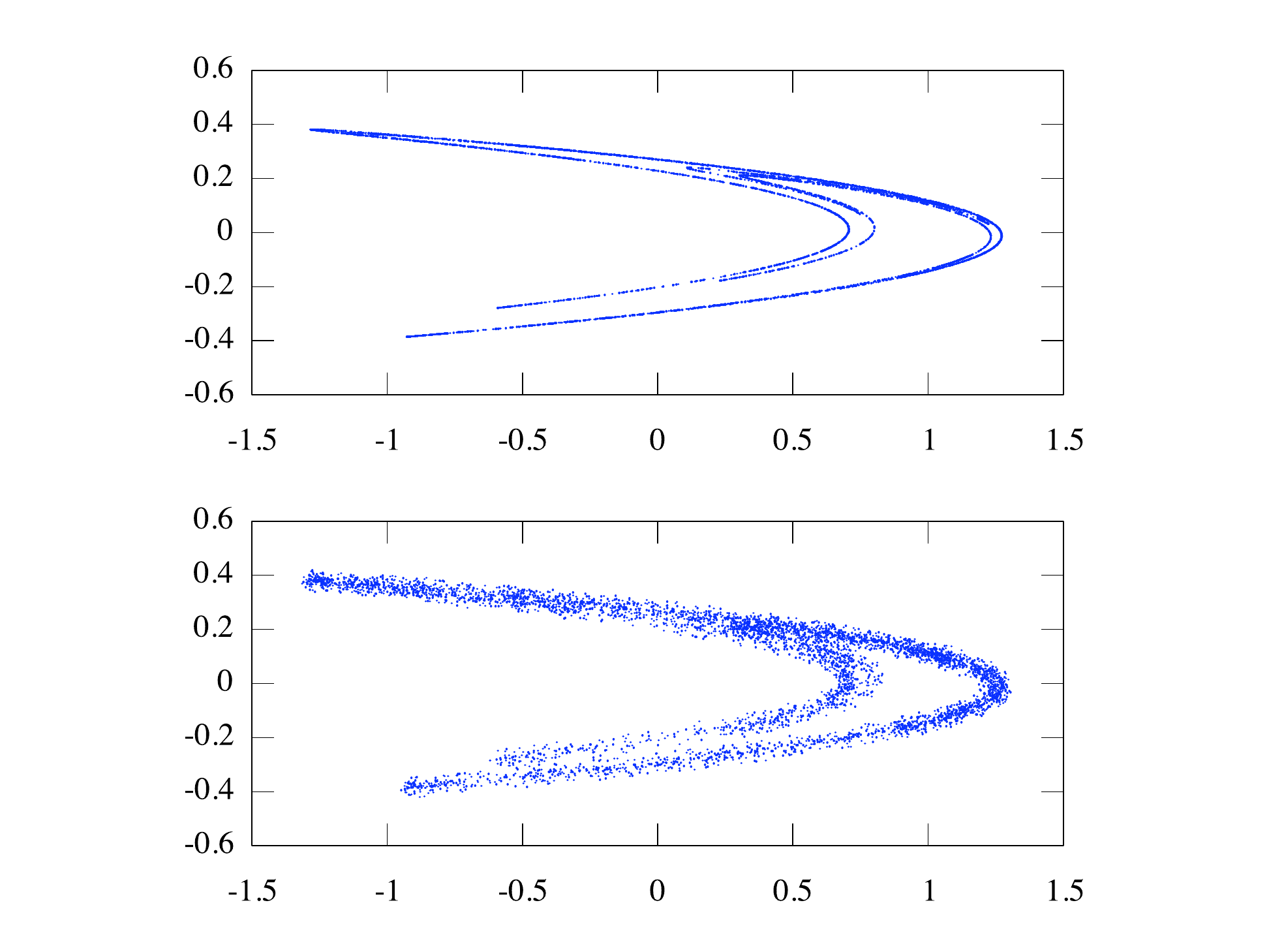}
\caption{Top: A simulation of the H\'{e}non map for the classical pa\-ra\-me\-ters a=1.4 and b=0.3.
Bottom: Simulation of the observed H\'{e}non map with observational noise whose magnitude is uniformly bounded by $\varepsilon=0.04$.}
\end{figure}

\begin{ex}\label{ex5}
The Manneville-Pomeau map is an example of an expansive map, except for a point where the slope is equal to 1 (neutral fixed point). Consider $X=[0,1]$, and for the sake of definiteness take
\[
T_{\alpha}(x) =
\begin{cases}
x+2^{\alpha}x^{1+\alpha} & \mbox{ if }\ x\in[0,1/2)\\
2x-1 	& \mbox{ if }\ x\in[1/2,1),
\end{cases}
\]
where $\alpha\in(0,1)$ is a parameter. It is well known that there exists an absolutely continuous invariant probability measure $d\mu(x)=h(x)dx$ and $h(x)\sim x^{-\alpha}$ when $x\to0$. The observed sequence is defined by $y_{i}= T^{i}_{\alpha}(x) +\varepsilon \xi_{i}$. The random variables $\xi_{i}$ are uniformly distributed in $X$. One identifies the $[0,1]$ with the unit circle to avoid leaks.
\end{ex}
\vspace{0.1cm}

\subsection{Concentration inequalities} 
Let $X$ be a metric space. For any function of $n$ variables $K: X^{n}\to\rr$, and for each $j$, $0\leq j \leq n-1$, let
\begin{equation*}
\Lip_{j}(K) := \sup_{x_{0},\ldots,x_{n-1}}\sup_{x_{j}\neq x'_{j}}\frac{\lvert K(x_{0},\ldots,x_{j},\ldots,x_{n-1})-K(x_{0},\ldots,x'_{j},\ldots,x_{n-1}) \rvert}{ d(x_{j},x'_{j})}.
\end{equation*}
We say that $K$ is \textsl{separately Lipschitz} if, for all $0\leq j \leq n-1$, $\Lip_{j}(K)$ is finite.\newline

Now, we may state the following definition.

\begin{defi}
The stochastic process $(Y_{n})$ taking values on $X$ satisfies an \textsl{exponential concentration inequality} if there exists a constant $C>0$ such that, for any separately Lipschitz function  $K$ of $n$ variables, one has
\begin{equation}\label{exp-ineq}
\ee\left( e^{K(Y_{0},\ldots,Y_{n-1}) - \ee(K(Y_{0},\ldots,Y_{n-1}))}\right)\leq e^{C\sum_{j=0}^{n-1}\Lip_{j}(K)^{2}}.
\end{equation}
\end{defi}
Notice that the constant $C$ depends only on $T$, but neither on $K$ nor on $n$.\newline

A weaker inequality is given by the following definition.

\begin{defi}
The stochastic process $(Y_{n})$ taking values on $X$ satisfies a \textsl{polynomial concentration inequality} with moment $q\geq2$ if there exists a constant $C_{q}>0$ such that, for any separately Lipschitz function $K$ of $n$ variables, one has
\begin{equation}\label{poly-ineq}
\ee\left( \lvert K(Y_{0},\ldots,Y_{n-1}) - \ee(K(Y_{0},\ldots,Y_{n-1}))\rvert^{q} \right)\leq C_{q}\left(\sum_{j=0}^{n-1}\Lip_{j}(K)^{2}\right)^{q/2}.
\end{equation}
\end{defi}
As in the previous definition the constant $C_{q}$ does not depend neither on $K$ nor on $n$.

\begin{rem}
When $q=2$, we have a bound for the variance of $K(Y_{0},\ldots,Y_{n})$.
\end{rem}

\begin{rem}\label{remark-iid}
If $(Y_{n})$ is a bounded i.i.d. process then it satisfies \eqref{exp-ineq} (see e.g. \cite{Led}). It also satisfies \eqref{poly-ineq} for all $q\geq2$, see e.g. \cite{BBL} for more details.
\end{rem}

These concentration inequalities allow us to obtain estimates on the deviation pro\-ba\-bi\-li\-ties of the observable $K$ from its expected value.

\begin{prop}
If the process $(Y_{n})$ satisfies the exponential concentration inequality \eqref{exp-ineq} then for all $t>0$ and for all $n\geq1$,
\begin{equation}\label{exp-prob-deviation}
\pp\left\{\lvert K(Y_{0},\ldots,Y_{n-1}) - \ee(K(Y_{0},\ldots,Y_{n-1})) \rvert>t\right\} \leq 2 e^{\frac{-t^{2}}{4C\sum_{j=0}^{n-1}\Lip_{j}(K)^{2} }}.
\end{equation}
If the process satisfies the polynomial concentration inequality \eqref{poly-ineq} for some $q\geq2$, then we have that for all $t>0$ and for all $n\geq1$,
\begin{equation}\label{poly-prob-deviation}
\pp\left\{\lvert K(Y_{0},\ldots,Y_{n-1})-\ee(K(Y_{0},\ldots,Y_{n-1})) \rvert >t\right\} \leq  \frac{C_{q}}{t^{q}}\Big(\sum_{j=0}^{n-1}\Lip_{j}(K)^{2}\Big)^{q/2}.
\end{equation}
\end{prop}

The inequality \eqref{exp-prob-deviation} follows from the basic inequality $\pp(Z>t) \leq e^{-\lambda t}\ee(e^{\lambda Z})$ with $\lambda >0$ applied to $Z = K(Y_{0},\ldots,Y_{n-1})-\ee(K(Y_{0},\ldots,Y_{n-1}))$, using the exponential concentration ine\-qua\-li\-ty \eqref{exp-ineq} and optimizing over $\lambda$. The inequality \eqref{poly-prob-deviation} follows easily from \eqref{poly-ineq} and the Markov inequality (see \cite{Ch} for details).

It has been proven that a dynamical system modeled by a Young tower with exponential tails satisfies the exponential concentration inequality  \cite{ChGo}. The systems in the examples from \ref{ex1} to \ref{ex4} are included in that framework. The example \ref{ex5} satisfies the polynomial concentration inequality with moment $q<\frac{2}{\alpha}-2$ for $\alpha\in(0,1/2)$, which is the parameter of the map (see \cite{ChGo} for full details).

\section{Main theorem \& corollary}

Let us first introduce some notations.
We recall that $P$ is the common distribution of the random variables $\xi_{i}$. 
The expected value with respect to a measure $\nu$ is denoted by $\ee_{\nu}$. 
Recall the expression \eqref{marginal} for the measure $\mu_{n}$. 
Hence in particular
\begin{align*}
\ee_{\mu_{n}}(K) =& \int\cdots\int K(x_{0},\ldots,x_{n-1})\dd\mu_{n}(x_{0},\ldots,x_{n-1})\\
=& \int K(x,\ldots, T^{n-1}x)\dd\mu(x).
\end{align*}
Next, we denote by $\mup$ the product of the measures $\mu_{n}$ and $P^{n}$, where $P^{n}$ stands for $P\otimes\cdots\otimes P$ ($n$ times).
The expected value of $K(y_{0}, \ldots,y_{n-1})$ is denoted by
\[
\ee_{\mup}(K) := \int K(x + \varepsilon\xi_{0},\ldots, T^{n-1}x+\varepsilon\xi_{n-1})\dd\mu(x)\dd{P}(\xi_{0})\cdots\dd{P}(\xi_{n-1}).
\]

Our main result is the following.

\begin{thm}\label{PerturbedIneq}
If the original system $(X,T,\mu)$ sa\-tis\-fies the exponential inequality \eqref{exp-ineq}, then the observed system satisfies an exponential concentration inequality. For any $n\geq1$, it is given by
\begin{equation}\label{exp-ineq-pert}
\ee_{\mup}\left( e^{K(y_{0},\ldots,y_{n-1})-\ee_{\mup}(K(y_{0},\ldots,y_{n-1}))}\right)\leq e^{D(1+\varepsilon^{2})\sum_{j=0}^{n-1}\Lip_{j}(K)^{2}},
\end{equation}
Furthermore, if the system $(X,T,\mu)$ satisfies the polynomial concentration  inequality \eqref{poly-ineq} with moment $q\geq2$, then the observed system satisfies a polynomial concentration inequality with the same moment. For any $n\geq1$, it is given by
\begin{equation}\label{poly-ineq-pert}
\begin{split}
\ee_{\mup}\left( \left\lvert K(y_{0},\ldots,y_{n-1})-\ee_{\mup}(K(y_{0},\ldots,y_{n-1}))\right\rvert^{q} \right) \leq D_{q}(1+\varepsilon)^{q}\Big(\sum_{j=0}^{n-1}\Lip_{j}(K)^{2}\Big)^{q/2}.
\end{split}
\end{equation}
\end{thm}

Observe that one recovers the corresponding concentration inequalities for the original dynamical system when $\varepsilon$ vanishes.

\begin{rem}
Our proof works provided the noise process satisfies a concentration inequality (see Remark \ref{remark-iid}).  We have stated the result in the particular case of i.i.d. noise because it is reasonable to model the observational perturbations in this manner. Nevertheless, one can slightly modify the proof to get the result valid for correlated perturbations.
\end{rem}

\begin{proof}[Proof of theorem \ref{PerturbedIneq}]
First let us fix the noise $\{\xi_{j}\}$ and let $\overline{\xi} := (\xi_{0}, \xi_{1}, \ldots,\xi_{n-1})$. Introduce the auxiliary observable
\[
\Kxi(x_{0},\ldots,x_{n-1}) := K(x_{0} + \varepsilon\xi_{0}, \ldots,x_{n-1}+\varepsilon\xi_{n-1}).
\]

Since the noise is fixed, it is easy to see that $\Lip_{j}(\Kxi)=\Lip_{j}(K)$ for all $j$. 

\noindent
Notice that $\Kxi(x,\ldots, T^{n-1}x)= K(x+\varepsilon\xi_{0},\ldots, T^{n-1}x+\varepsilon\xi_{n-1}) = K(y_{0},\ldots,y_{n-1})$. Next we define the observable $F(\xi_{0}, \ldots, \xi_{n-1})$ of $n$ variables on the noise, as follows,
\[
F(\xi_{0}, \ldots, \xi_{n-1}):= \ee_{\mu_{n}}(\Kxi(x,\ldots, T^{n-1}x)).
\]
Observe that, $\Lip_{j}(F) \leq \varepsilon\Lip_{j}(K)$. 

\noindent
Now we prove inequality \eqref{exp-ineq-pert}. Observe that is equivalent to prove the inequality for
\begin{equation*}
\ee_{\mup}\left(e^{\Kxi(x,\ldots,T^{n-1}x)-\ee_{\mup}(\Kxi(x,\ldots,T^{n-1}x))} \right). 
\end{equation*}
Adding and subtracting $\ee_{\mup}(\Kxi(x, \ldots, T^{n-1}x))$ and using the independence between the noise and the dynamical system, we obtain that the expression above is equal to
\[
\ee_{\mu_{n}}\left(e^{\Kxi(x, \ldots,T^{n-1}x)- \ee_{\mu_{n}}(\Kxi(x,\ldots,T^{n-1}x))} \right)\ee_{P^{n}}\left(e^{F(\xi_{0}, \ldots, \xi_{n-1})-\ee_{P^{n}}(F(\xi_{0}, \ldots, \xi_{n-1}))} \right).
\]
Since in particular, i.i.d. bounded processes satisfy the exponential concentration inequality (see Remark \ref{remark-iid} above), we may apply \eqref{exp-ineq} to the dynamical system and the noise, yielding
\begin{equation*}
\begin{split}
\ee_{\mu_{n}}\left(e^{\Kxi(x, \ldots,T^{n-1}x)- \ee_{\mu_{n}}(\Kxi(x,\ldots,T^{n-1}x))} \right)\ee_{P^{n}}\left(e^{F(\xi_{0}, \ldots, \xi_{n-1})-\ee_{P^{n}}(F(\xi_{0}, \ldots, \xi_{n-1}))} \right)\\
\leq e^{C\sum_{j=0}^{n-1}\Lip_{j}(\Kxi)^{2}}e^{C'\varepsilon^{2}\sum_{j=0}^{n-1}\Lip_{j}(F)^{2}} \leq e^{D(1+\varepsilon^{2})\sum_{j=0}^{n-1}\Lip_{j}(K)^{2}},
\end{split}
\end{equation*}
where $D := \max\{C,C'\}$. 

\noindent
Next, we prove inequality \eqref{poly-ineq-pert} similarly. We use the binomial expansion after the triangle inequality with $\ee_{\mu_{n}}(\Kxi(x,\ldots,T^{n-1}x))$. Using the independence between the noise and the dynamics, we get
\begin{equation}\label{sep}
\begin{split}
\ee_{\mup}( \lvert K(y_{0},\ldots, y_{n-1}) - \ee_{\mup}(K(y_{0},\ldots,y_{n-1}))\rvert^{q}) \hspace{4cm}\\
\leq\sum_{p=0}^{q}\Big(\begin{array}{c} q\\p \end{array}\Big)\ee_{\mu_{n}}(\lvert \Kxi(x,\ldots,T^{n-1}x) - \ee_{\mu_{n}}(\Kxi(x,\ldots,T^{n-1}x))\rvert^{p})\times\\
\ee_{P^{n}}\left(\lvert F(\xi_{0},\ldots,\xi_{n-1})-\ee_{P^{n}}(F(\xi_{0},\ldots,\xi_{n-1}))\rvert^{q-p} \right).
\end{split}
\end{equation}

\noindent
We proceed carefully using the polynomial concentration inequality. The terms corresponding to $p=1$ and $p=q-1$ have to be treated separately. For the rest we obtain the bound
\begin{equation*}
\sum_{\substack{ p=0\\ p\neq1,q-1}}^{q} \binom{q}{p}
C_{p}\Big(\sum_{j=0}^{n-1}\Lip_{j}(K)^{2}\Big)^{p/2}\times C'_{q-p}\Big(\varepsilon^{2}\sum_{j=0}^{n-1}\Lip_{j}(K)^{2}\Big)^{\frac{q-p}{2}}.
\end{equation*}
For the case $p=1$, we apply Cauchy-Schwarz inequality and \eqref{poly-ineq} for $q=2$ to get
\[
\ee_{\mu_{n}}\left(\lvert \Kxi(x,\ldots,T^{n-1}x)-\ee_{\mu_{n}}(\Kxi(x,\ldots,T^{n-1}x))\rvert\right)\leq \sqrt{C_{2}}\Big(\sum_{j=0}^{n-1}\Lip_{j}(K)^{2}\Big)^{1/2}.
\]
If $q=2$, we proceed in the same way for the second factor in the right hand side of \eqref{sep}.
The case $p=q-1$ is treated similarly. Finally, putting this together and choosing adequately the constant $D_{q}$ we obtain the desired bound.
\end{proof}

Next we obtain an estimate of deviation probability of the observable $K$ from its expected value.

\begin{coro} 
If the system $(X,T,\mu)$ satisfies the exponential concentration inequality, then for the observed system $\{y_{i}\}$, for every $t>0$ and for any $n\geq1$ we have,
\begin{equation}\label{exp-deviation-pert}
\mup\big( \lvert K(y_{0},\ldots,y_{n-1})-\ee_{\mup}(K)\rvert \geq t \big)\leq 2\exp\left(\frac{-t^{2}}{4D(1+\varepsilon^{2})\sum_{j=0}^{n-1}\Lip_{j}(K)^{2}}\right).
\end{equation}
If the system $(X,T,\mu)$ satisfies the polynomial concentration inequality with moment $q\geq2$, then the observed system satisfies for every $t>0$ and for any $n\geq1$,
\begin{equation}\label{poly-deviation-pert}
\begin{split}
\mup\big(\lvert K(y_{0},\ldots,y_{n-1}) - \ee_{\mup}(K)\rvert >t\big) \leq 
\frac{D_{q}}{t^{q}}(1+\varepsilon)^{q}\left(\sum_{j=0}^{n-1}\Lip_{j}(K)^{2}\right)^{q/2}.
\end{split}
\end{equation}
\end{coro}

The proof is straightforward and left to the reader.

\section{Applications}

\subsection{Dynamical systems}
Concentration inequalities are available for the class of non-uniformly hyperbolic dynamical systems modeled by Young towers (\cite{ChGo}). Actually, systems with exponential tails satisfy an exponential concentration inequality and if the tails are polynomial then the system satisfies a polynomial concentration inequality. The examples given in section 2 are included in that class of dynamical systems. We refer the interested reader to \cite{Y1} and \cite{Y2} for more details on systems modeled by Young towers. 
Here we consider dynamical systems satisfying either the exponential or the polynomial concentration inequality. We apply our result of concentration in the setting of observed systems to empirical estimators of the auto-covariance function, the em\-pi\-ri\-cal measure, the kernel density estimator and the correlation dimension. 

\subsection{Auto-covariance function}

Consider the dynamical system $(X,T,\mu)$ and a square integrable observable $f:X\to\rr$. Assume that $f$ is such that $\int f\dd\mu=0$. We remind that the auto-covariance function of $f$ is given by
\begin{equation*}
\mathrm{Cov}(k):= \int f(x)f(T^{k}x)\dd\mu(x).
\end{equation*}
In practice, one has a finite number of iterates of some $\mu$-typical initial condition $x$, thus, what we may easily obtain from the data is the empirical estimator of the auto-covariance function: 
\begin{equation*}
\Cemp(k) := \frac{1}{n}\sum_{i=0}^{n-1}f(T^{i}x)f(T^{i+k}x).
\end{equation*}
From Birkhoff's ergodic theorem it follows that $\mathrm{Cov}(k)=\lim_{n\to\infty}\Cemp(k)$ $\mu$-almost surely. Observe that the expected value of the estimator $\Cemp(k)$ is exactly $\mathrm{Cov}(k)$.

The following result gives us \textit{a priori} theoretical bounds to the fluctuations of the estimator $\Cemp$ around $\mathrm{Cov}$ for every $n$. This result can be found in \cite{ChGo}, here we include it for the sake of completeness.

\begin{prop}\label{Prop-Concen-Cov}
Let $\mathrm{Cov}(k)$ and $\Cemp(k)$ be defined as above. If the dynamical system $(X,T,\mu)$ satisfies the exponential concentration inequality \eqref{exp-ineq} then for all $t>0$ and any integer $n\geq1$ we have
\begin{equation*}
\mu\left( \left\lvert \Cemp(k) - \mathrm{Cov}(k) \right\rvert > t \right) \leq 2\exp\left(\frac{-t^{2}}{16Ca_{f}^{2} }\left(\frac{n^{2}}{n+k}\right) \right),
\end{equation*}
where $a_{f} = \Lip(f)\lVert{f}\rVert_{\infty}$ and $C$ is the constant appearing in \eqref{exp-ineq}.

\noindent If the system satisfies the polynomial concentration inequality \eqref{poly-ineq} with moment $q\geq2$, then for all $t>0$ and any integer $n\geq1$ we have
\begin{equation*}
\mu\left( \left\lvert \Cemp(k)-\mathrm{Cov}(k)\right\rvert>t \right)\leq C_{q}\left(\frac{2a_{f}}{t}\right)^{q}\left(\frac{n+k}{n^{2}}\right)^{q/2},
\end{equation*}
where $C_{q}$ is the constant appearing in \eqref{poly-ineq}.
\end{prop}

\begin{proof}
Consider the following observable of $n+k$ variables,
\[
K(z_{0},\ldots,z_{n+k-1}):= \frac{1}{n}\sum_{i=0}^{n-1}f(z_{i})f(z_{i+k}).
\]
In order to estimate the Lipschitz constant of $K$, consider $0\leq l\leq n+k-1$ and replace the value $z_{l}$ with $z'_{l}$. Note that the absolute value of the difference between $K(z_{0},\ldots,z_{l},\ldots,z_{n+k-1})$ and
$K(z_{0},\ldots,z'_{l},\ldots,z_{n+k-1})$ is less than or equal to 
\[
\frac{1}{n}\left\lvert f(z_{l-k})f(z_{l}) + f(z_{l})f(z_{l+k})-f(z_{l-k})f(z'_{l})-f(z'_{l})f(z_{l+k})\right\rvert,
\] 
and so for every index $l$, we have that

\begin{equation*}
\Lip_{l}(K) \leq \sup_{z_{0},\ldots,z_{n+k-1}}\sup_{z_{l}\neq z'_{l}} \frac{1}{n}\frac{\lvert(f(z_{l})-f(z'_{l}))(f(z_{l-k})+f(z_{l+k})) \rvert}{d(z_{l}, z'_{l})}\leq  \frac{2}{n}\Lip(f)\lVert f\rVert_{\infty}.
\end{equation*}

Next, if the exponential inequality holds, we use \eqref{exp-prob-deviation} to obtain
\begin{align*}
\mu\left(\Cemp(k)-\mathrm{Cov}(k)>t\right) \leq& \exp\left(\frac{-t^{2}}{16C\Lip(f)^{2}\lVert f\rVert^{2}_{\infty}}\left( \frac{n^{2}}{n+k}\right) \right).
\end{align*}
Applying similarly the inequality to the function $-K$, we get the result by a union bound. The polynomial case follows from inequality \eqref{poly-prob-deviation}.
\end{proof}

\subsubsection{Auto-covariance function for observed systems}
Let us consider the observed orbit $y_{0},\ldots,y_{n-1}$. Define the observed empirical estimator of the auto-covariance function as follows
\begin{equation}\label{cov-funct-pert}
\Cpert(k) := \frac{1}{n}\sum_{i=0}^{n-1}f(y_{i})f(y_{i+k}).
\end{equation}
We are interested in quantifying the influence of noise on the correlation. We provide a bound on the probability of the deviation of the observed empirical estimator from the covariance function.

\begin{thm}
Let $\Cpert(k)$ be given by \eqref{cov-funct-pert}. If the dynamical system $(X,T,\mu)$ satisfies the exponential inequality \eqref{exp-ineq} then for all $t>0$ and for any integer $n\geq1$ we have
\begin{equation*}
\begin{split}
\mup\left(\left\lvert \Cpert(k) -\mathrm{Cov}(k)\right\rvert > t + 2a_{f}\varepsilon \right)\leq
2\exp\left( \frac{-t^{2}}{64Da_{f}^{2}(1+\varepsilon^{2})}\left(\frac{n^{2}}{n+k}\right)\right)\\ + 2\exp\left(\frac{-t^{2}}{16Ca_{f}^{2} }\left(\frac{n^{2}}{n+k}\right)\right),
\end{split}
\end{equation*}
where $a_{f}=\Lip(f)\lVert f\rVert_{\infty}$, $C$ and $D$ are the constants appearing in \eqref{exp-ineq} and \eqref{exp-ineq-pert} respectively. 
If the system satisfies the polynomial inequality with moment $q\geq2$, then for all $t>0$ and any integer $n\geq1$ we have
\begin{equation*}
\mup\left( \left\lvert \Cpert(k) - \mathrm{Cov}(k)\right\rvert >t + 2a_{f}\varepsilon \right)\leq 
\left(2^{q}D_{q}(1+\varepsilon)^{q}+C_{q}\right)\left(\frac{2a_{f}}{t}\right)^{q}\left(\frac{n+k}{n^{2}}\right)^{q/2},
\end{equation*}
where $C_{q}$ and $D_{q}$ are the constants appearing in \eqref{poly-ineq} and \eqref{poly-ineq-pert} respectively.
\end{thm}

\begin{proof}
To prove this assertion we will use an estimate of
\begin{equation*}
\mup\left( \left\lvert \Cpert(k) -\Cemp(k) \right\rvert > t + \ee_{\mup}\left(\left\lvert\Cpert(k) -\Cemp(k)\right\rvert\right) \right).
\end{equation*}
First let us write $x_{i}:=T^{i}x$, and observe that by adding and subtracting 
$f(x_{i}+\varepsilon\xi_{i})f(x_{i+k})$, the quantity $\lvert\Cpert(k) -\Cemp(k)\rvert$ is less than or equal to
\begin{equation*}
\frac{1}{n}\sum_{i=0}^{n-1}\left\lvert f(x_{i}+\varepsilon\xi_{i})[f(x_{i+k}+\varepsilon\xi_{i+k})-f(x_{i+k})]
+[f(x_{i}+\varepsilon\xi_{i})-f(x_{i})]f(x_{i+k})\right\rvert ,
\end{equation*}
which leads us to the following estimate,
\begin{equation}\label{expect-CpertCemp}
\ee_{\mup}\left(\left\lvert\Cpert(k) - \Cemp(k)\right\rvert\right) \leq 2\varepsilon\Lip(f)\lVert f\rVert_{\infty}.
\end{equation}
For a given realization of the noise $\{ e_{i}\}$, consider the following observable of $n+k$ variables
\[
K(z_{0},\ldots,z_{n+k-1}) := \frac{1}{n}\sum_{i=0}^{n-1}\left( f(z_{i}+\varepsilon e_{i})f(z_{i+k}+\varepsilon e_{i+k}) - f(z_{i})f(z_{i+k})\right).
\]
For every $0\leq l\leq n-1$, one can easily obtain that
\begin{equation*}
\Lip_{l}(K)\leq \frac{4}{n}\Lip(f)\lVert f\rVert_{\infty}.
\end{equation*}

In the exponential case, from the inequality \eqref{exp-deviation-pert} and the bound \eqref{expect-CpertCemp} on the expected value of  $K$, we obtain that

\begin{equation*}
\mup\left(\left\lvert \Cpert(k)-\Cemp(k)\right\rvert> t +2\varepsilon a_{f}\right)\leq
2\exp\left( \frac{-t^{2}}{64Da_{f}^{2}(1+\varepsilon^{2})}\left(\frac{n^{2}}{n+k}\right) \right).
\end{equation*}
Using proposition \ref{Prop-Concen-Cov}, a union bound and an adequate rescaling, we get the result. In order to prove the polynomial inequality, proceed similarly applying \eqref{poly-deviation-pert}.
\end{proof}

\subsection{Empirical measure}

The empirical measure of a sample $x_{0},\ldots,x_{n-1}$ is given by
\begin{equation*}
\E := \frac{1}{n}\sum_{i=0}^{n-1}\delta_{x_{i}},
\end{equation*}
where $\delta_{x}$ denotes the Dirac measure at $x$. If the given sample $x_{0},\ldots,x_{n-1}$ is the sequence $x,\ldots,T^{n-1}x$ for a $\mu$-typical $x\in X$, then from Birkhoff's ergodic theorem it follows that the sequence of random measures $\{\E\}$ converges weakly to the $T$-invariant measure $\mu$, almost surely. 

Consider the observed itinerary $y_{0}, \ldots,y_{n-1}$ and define the observed empirical measure by
\begin{equation*}
\Epert := \frac{1}{n}\sum_{i=0}^{n-1}\delta_{y_{i}}.
\end{equation*}
Observe that this measure is well defined on $X$. Again Birkhoff's ergodic theorem implies that almost surely
\[
\lim_{n\to\infty}\frac{1}{n}\sum_{i=0}^{n-1}g(y_{i}) = \int\int g(x+\xi)\dd\mu(x)\dd{P}(\xi),
\]
for every continuous function $g$. More precisely, this convergence holds for a set of $\mu$-measure one of initial conditions for the dynamical system $(X,T)$ and a set of measure one of noise realizations $(\xi_{i})$ with respect to the product measure $P^{\nn}$.

We want to estimate the speed of convergence of the observed empirical measure. For that purpose, we chose the Kantorovich distance on the set of probability measures, which is defined by
\begin{equation*}
\kappa(\mu,\nu) := \sup_{g\in\mathcal{L}} \int g \dd\mu - \int g \dd\nu,
\end{equation*}
where $\mu$ and $\nu$ are two probability measures on $X$ and $\mathcal{L}$ denotes the space of all real-valued Lipschitz functions on $X$ with Lipschitz constant at most one.

Now, we study the fluctuations of the Kantorovich distance of the observed empirical measure to the measure $\mu$, around its expected value. The statement is the following.

\begin{prop}\label{Fluctua-Emp-Meas}
If the system $(X,T,\mu)$ satisfies the exponential concentration inequality \eqref{exp-ineq}, then for all $t>0$ and any integer $n\geq1$,
\[
\mup\left( \kappa(\Epert,\mu)>t +\ee_{\mup}\big(\kappa(\Epert,\mu)\big)\right)\leq e^{-\frac{t^{2}n}{4D(1+\varepsilon^{2})}}.
\]
If the system satisfies the polynomial concentration inequality \eqref{poly-ineq} with moment $q\geq2$, then for all $t>0$ and any integer $n\geq1$,
\[
\mup\left( \kappa(\Epert,\mu) >t + \ee_{\mup}\big(\kappa(\Epert,\mu)\big) \right)\leq
 \frac{D_{q}(1+\varepsilon)^{q}}{{t}^{q}}\frac{1}{n^{q/2}}.
\]
\end{prop}

Using the following separately Lipschitz function of $n$ variables,
\[
K(z_{0},\ldots,z_{n-1}) := \sup_{g\in\mathcal{L}}\left[\frac{1}{n}\sum_{i=0}^{n-1}g(z_{i})-\int g \dd\mu\right].
\]
It is easy to check that $\Lip_{j}(K)\leq \frac{1}{n}$, for every $j=0,\ldots,n-1$. The proposition follows from the concentration inequalities \eqref{exp-deviation-pert} and \eqref{poly-deviation-pert}. 

We are not able to obtain a sufficiently good estimate of $\ee_{\mup}\left(\kappa(\Epert,\mu)\right)$ in dimension larger than one, thus in the following we restrict ourselves to systems with $X\subset\rr$.

\begin{lemma}[\cite{ChCS}]\label{Bound-Expect}
Let $(X,T,\mu)$ be a dynamical system with $X\subset\rr$. If there exists a constant $c>0$ such that for every Lipschitz function $f$, the auto-covariance function $\mathrm{Cov}_{f}(k)$ satisfies that $\sum_{k=1}^{\infty}\lvert \mathrm{Cov}_{f}(k)\rvert\leq{c}\lVert{f}\rVert_{\Lip}^{2}$, then there exists a constant $B$ such that for all $n\geq1$
\begin{equation*}
\ee_{\mu_{n}}\left(\kappa(\E,\mu)\right) \leq \frac{B}{n^{1/4}}.
\end{equation*}
\end{lemma}

The proof of the preceding lemma is found in \cite[Section 5]{ChCS}. It relies in the fact that in dimension one, it is possible to rewrite the Kantorovich distance using distribution functions. Then by an adequate Lipschitz approximation of the distribution function, the estimate bound follows from the summability condition on the auto-covariance function.

As a consequence of proposition \ref{Fluctua-Emp-Meas} and the previous lemma, we obtain the following result.

\begin{thm}
Assume that the system $(X,T,\mu)$ satisfies the assumptions of lemma \ref{Bound-Expect}. Let $\Epert$ be the observed empirical measure. If the system satisfies the exponential inequality \eqref{exp-ineq} then for all $t>0$ and for all $n\geq1$ we have that
\begin{equation*}
\mup\left( \kappa(\Epert,\mu)> \frac{t +B}{n^{1/4}}+\varepsilon \right)\leq e^{-\frac{t^{2}\sqrt{n}}{4D(1+\varepsilon^{2})}}.
\end{equation*}
If the system satisfies the polynomial inequality \eqref{poly-ineq} with moment $q\geq2$, then for all $t>0$ and for all $n\geq1$ we obtain
\begin{equation*}
\mup\left( \kappa(\Epert,\mu)> \frac{t+B}{n^{1/4}}+ \varepsilon \right)\leq \frac{D_{q}(1+\varepsilon)^{q}}{t^{q}}\frac{1}{n^{q/4}}.
\end{equation*}
\end{thm}

\begin{proof}
Clearly $ \ee_{\mup}(\kappa(\Epert,\mu))\leq \ee_{\mup}(\kappa(\Epert,\E)) + \ee_{\mup}(\kappa(\E,\mu))$.
A straightforward estimation yields
\begin{align*}
\ee_{\mup}(\kappa(\Epert,\E))
\leq&\int \sup_{g\in\mathcal{L}}\left[ \frac{1}{n}\sum_{i=0}^{n-1} \Lip(g)\varepsilon\lVert \xi_{i}\rVert \right] \dd\mup \
\leq \ \varepsilon.
\end{align*}
We obviously have $\ee_{\mup}\left(\kappa(\E,\mu)\right)=\ee_{\mu_{n}}\left(\kappa(\E,\mu)\right)$. Using the exponential estimate of proposition \ref{Fluctua-Emp-Meas} and lemma \ref{Bound-Expect} we obtain, for any $t>0$,
\[
\mup\left( \kappa(\Epert,\mu)\geq t+\varepsilon +\frac{B}{n^{1/4}} \right)\leq \exp\left(\frac{-t^{2}n}{4D(1+\varepsilon^{2})} \right).
\]
Rescaling adequately we get the result. For the polynomial case, one uses the polynomial estimate of proposition \ref{Fluctua-Emp-Meas}.
\end{proof}

\subsection{Kernel density estimator for one-dimensional maps}
In this section we consider the system $(X,T,\mu)$ where $X$ is a bounded subset of $\rr$. We assume the measure $\mu$ to be absolutely continuous with density $h$. For a given trajectory of a randomly chosen initial condition $x$ (according to $\mu$), the empirical density estimator is defined by,
\begin{equation*}
\h(x; s) :=\frac{1}{n\alpha_{n}}\sum_{j=0}^{n-1}\psi\left(\frac{s-T^{j}x}{\alpha_{n}}\right),
\end{equation*}
where $\alpha_{n}\to0$ and $n\alpha_{n}\to\infty$ as $n$ diverges. The kernel $\psi$ is a bounded and non-negative Lipschitz function with bounded support and it satisfies $\int\psi(s)\dd{s}=1$. We shall use the following hypothesis.

\begin{hypo}\label{hypo-density}
The probability density $h$ satisfies
\begin{equation*}
\int\lvert h(s) - h(s-\sigma)\rvert \dd s\leq C'\lvert \sigma \rvert^{\beta}
\end{equation*}
for some constants $C'>0$ and $\beta>0$ and for every $\sigma\in\rr$.
\end{hypo}

This assumption is indeed valid for maps on the interval satisfying the axioms of Young towers with exponential tails (see \cite[Appendix C]{ChCS}). For convenience, we present the following result on the $L^{1}$ convergence of the density estimator (\cite{ChGo}).

\begin{prop}
Let $\psi$ be a kernel defined as above. If the system $(X,T,\mu)$ satisfies the exponential concentration inequality \eqref{exp-ineq} and the hypothesis \ref{hypo-density}, then there exist a constant $C_{\psi}>0$ such that for any integer $n\geq1$ and every $t>C_{\psi}\left(\alpha_{n}^{\beta}+\frac{1}{\sqrt{n}\alpha_{n}^{2}}\right)$, we have
\begin{equation*}
\mu\left( \int\left\lvert \h(x;s)-h(s)\right\rvert \dd{s} >t\right)\leq e^{-\frac{n\alpha_{n}^{4}t^{2}}{4C\Lip(\psi)^{2}}}.
\end{equation*}
Under the same conditions above, if the system satisfies the polynomial concentration inequality \eqref{poly-ineq} for some $q\geq2$, then  for any integer $n\geq1$ and every $t>C_{\psi}\left(\alpha_{n}^{\beta}+\frac{1}{\sqrt{n}\alpha_{n}^{2}}\right)$, we obtain,
\begin{equation*}
\mu\left( \int\left\lvert \h(x;s)-h(s)\right\rvert \dd{s} >t\right)\leq \frac{C_{q}}{t^{q}}\left(\frac{\Lip(\psi)}{\sqrt{n}\alpha_{n}^{2}} \right)^{q}.
\end{equation*}
The parameter $\beta$ is the same constant appearing in the hypothesis \ref{hypo-density}.
\end{prop}

For the proof of this statement see \cite{ChGo} or Theorem 6.1 in \cite{ChCS}.

\subsubsection{Kernel density estimator for observed maps on the circle}
In order to avoid `leaking' problems, now we assume $X=\mathbb{S}^{1}$.
Given the observed sequence $\{y_{j}\}$, let us define the observed empirical density estimator by
\begin{equation*}
\hpert(y_{0},\ldots,y_{n-1};s) : = \frac{1}{n\alpha_{n}}\sum_{j=0}^{n-1}\psi\left(\frac{s-y_{j}}{\alpha_{n}}\right).
\end{equation*}

Our result is the following.

\begin{thm}
If $(X,T,\mu)$ satisfies the hypothesis \ref{hypo-density} and the exponential concentration inequality, then there exists a constant $C_{\psi}>0$ such that, for all $t>C_{\psi}\left(\alpha_{n}^{\beta} + \frac{1}{\sqrt{n}\alpha_{n}^{2}}\right)$ and for any integer $n\geq1$,
\begin{equation*}
\mup\left( \int\left\lvert \hpert(y_{0},\ldots,y_{n-1};s) - h(s)\right\rvert \dd{s}> t + \Lip(\psi)\frac{\varepsilon}{\alpha_{n}^{2}}\right) \leq \exp\left(-\frac{n\alpha_{n}^{4}t^{2}}{R(1+\varepsilon^{2})}\right),
\end{equation*}
where $R:=4D\Lip(\psi)^{2}$.

\noindent If the system satisfies the hypothesis \ref{hypo-density} and the polynomial concentration inequality, then for all $t>C_{\psi}\left(\alpha_{n}^{\beta} + \frac{1}{\sqrt{n}\alpha_{n}^{2}}\right)$ and for any integer $n\geq1$, we have
\begin{equation*}
\mup\left( \int\left\lvert \hpert(y_{0},\ldots,y_{n-1};s) - h(s)\right\rvert \dd{s}> t + \Lip(\psi)\frac{\varepsilon}{\alpha_{n}^{2}}\right) \leq  D_{q}\left( \frac{(1+\varepsilon)\Lip(\psi)}{t\sqrt{n}\alpha_{n}^{2}}\right)^{q}.
\end{equation*}
The parameter $\beta$ is the same constant appearing as in the hypothesis \ref{hypo-density}.
\end{thm}

\begin{proof}
Consider the following observable of $n$ variables,
\[
K(z_{0},\ldots,z_{n-1}) := \int\Big\lvert \frac{1}{n\alpha_{n}}\sum_{j=0}^{n-1}\psi\left( \frac{s-z_{j}}{\alpha_{n}}\right) - h(s)\Big\rvert \dd{s}.
\]
It is straightforward to obtain that $\Lip_{l}(K) \leq\frac{\Lip(\psi)}{n\alpha_{n}^{2}}$,
for every $l=0,\ldots, n-1$. Next, we need to give an upper bound for the expected value of the observable $K$, first
\begin{align*}
\ee_{\mup}(K) \leq& 
\int \Big(\int \Big\lvert \frac{1}{n\alpha_{n}} \sum_{j=0}^{n-1}\left[\psi\Big(\frac{s-y_{j}}{\alpha_{n}}\Big)- \psi\left(\frac{s-x_{j}}{\alpha_{n}}\right)\right]\Big\rvert \dd{s}\Big)\dd\mup \\
& \hspace{1.3cm}+ \int\Big(\int \Big\lvert \frac{1}{n\alpha_{n}}\sum_{j=0}^{n-1}\psi\Big(\frac{s-x_{j}}{\alpha_{n}}\Big)-h(s)\Big\rvert \dd{s}\Big)\dd\mu_{n}.
\end{align*}
Subsequently we proceed on each part. For the first one we get
\begin{equation*}
\begin{split}
\int \Big(\int \Big\lvert \frac{1}{n\alpha_{n}} \sum_{j=0}^{n-1}\left[\psi\Big(\frac{s-y_{j}}{\alpha_{n}}\Big)-\psi\Big(\frac{s-x_{j}}{\alpha_{n}}\Big)\right]\Big\rvert \dd{s}\Big)d\mup\\
 \leq  \int \Big(\frac{1}{n\alpha_{n}}\sum_{j=0}^{n-1}\frac{\Lip(\psi)\varepsilon}{\alpha_{n}}\Big) \dd\mup &\ \leq \ \Lip(\psi)\frac{\varepsilon}{\alpha^{2}_{n}}.
\end{split}
\end{equation*}
For the second part, there exist some constant $C_{\psi}$ such that
\[
\int\Big(\int \Big\lvert \frac{1}{n\alpha_{n}}\sum_{j=0}^{n-1}\psi\Big(\frac{s-x_{j}}{\alpha_{n}}\Big)-h(s)\Big\rvert \dd{s}\Big)\dd\mu_{n} \leq C_{\psi}\left(\alpha_{n}^{\beta} + \frac{1}{\sqrt{n}\alpha_{n}^{2}}\right).
\]
The proof of this statement is found in \cite[Section 6]{ChCS}. We finish the proof applying \eqref{exp-deviation-pert} and \eqref{poly-deviation-pert}, respectively.
\end{proof}

\subsection{Correlation dimension}
The correlation dimension $d_{c} = d_{c}(\mu)$ of the measure $\mu$ is defined by
\[
d_{c} = \lim_{r\searrow0}\frac{\log{\int\mu(B_{r}(x))\dd\mu(x)}}{\log{r}},
\]
provided the limit exists. We denote by $\C(r)$ the \textsl{spatial correlation integral} which is defined by
\[
\C(r) = \int\mu(B_{r}(x))\dd\mu(x).
\]
As empirical estimator of $\C(r)$ we choose the following function of $n$ variables
\[
K_{n,r}(x_{0},\ldots,x_{n-1}) := \frac{1}{n^{2}}\sum_{i\neq j}H(r-d(x_{i},x_{j})),
\]
where $H$ is the Heaviside function. It has been proved (see e.g. \cite{Ser}) that
\[
\C(r) = \lim_{n\to\infty}K_{n,r}(x,\ldots, T^{n-1}x),
\]
$\mu$-almost surely at the continuity points of $\C(r)$. Next, given a $\mu$-typical initial condition, let us consider the observed sequence $y_{0},\ldots,y_{n-1}$, and define the estimator of $\C(r)$ for observed systems, as follows
\begin{equation*}
\K(y_{0},\ldots,y_{n-1}) := \frac{1}{n^{2}}\sum_{i\neq j}H(r-d(y_{i},y_{j})).
\end{equation*}

Since $\K(y_{0},\ldots,y_{n-1})$ is not a Lipschitz function we cannot apply directly concentration inequalities. The usual trick is to replace $H$ by a Lipschitz continuous function $\phi$ and then define the new estimator 
\begin{equation}\label{estim-pert-Corr}
\K^{\phi}(y_{0},\ldots,y_{n-1}):=\frac{1}{n^{2}}\sum_{i\neq j}\phi\left(1-\frac{d(y_{i},y_{j})}{r}\right).
\end{equation}

The result of this section is the following estimate on the variance of the estimator $\K^{\phi}$.

\begin{thm}
Let $\phi$ be a Lipschitz continuous function. Consider the observed trajectory $y_{0},\ldots,y_{n-1}$ and the function $\K^{\phi}(y_{0},\ldots,y_{n-1})$ given by \eqref{estim-pert-Corr}. If the system $(X,T,\mu)$ satisfies the polynomial concentration inequality with $q=2$, then for any integer $n\geq1$,
\begin{equation*}
\mathrm{Var}(\K^{\phi})\leq D_{2}\Lip(\phi)^{2}(1+\varepsilon)^{2}\frac{1}{r^{2}n},
\end{equation*}
where $\mathrm{Var}(Y) := \ee(Y^{2}) - \ee(Y)^{2}$ is the variance of $Y$.
\end{thm}

The proof follows the lines of section 4 in \cite{ChCS}, and by applying the inequality \eqref{poly-ineq-pert} with $q=2$ and noticing that $\Lip_{l}(\K^{\phi}) \leq \frac{\Lip(\phi)}{rn}$ for every $l=0,\ldots,n-1$.

\bibliographystyle{plain}
\bibliography{bib}

\begin{thebibliography}{10}

\bibitem{Aba}
Henry D.~I. Abarbanel.
\newblock {\em Analysis of observed chaotic data}.
\newblock Institute for Nonlinear Science. Springer-Verlag, New York, 1996.

\bibitem{Arn}
Ludwig Arnold.
\newblock {\em Random dynamical systems}.
\newblock Springer Monographs in Mathematics. Springer-Verlag, Berlin, 1998.

\bibitem{BY}
Michael Benedicks and Lai-Sang Young.
\newblock Sina\u\i-{B}owen-{R}uelle measures for certain {H}\'enon maps.
\newblock {\em Invent. Math.}, 112(3):541--576, 1993.

\bibitem{BBL}
St{\'e}phane Boucheron, Olivier Bousquet, G{\'a}bor Lugosi, and Pascal Massart.
\newblock Moment inequalities for functions of independent random variables.
\newblock {\em Ann. Probab.}, 33(2):514--560, 2005.

\bibitem{Ch}
J.-R. Chazottes.
\newblock Fluctuations of observables in dynamical systems: from limit theorems
  to concentration inequalities.
\newblock In {\em Nonlinear Dynamics: New Directions. Dedicated to Valentin
  Afraimovich on the occasion of his 65th birthday}. To appear, 2012.

\bibitem{ChCS}
J.-R. Chazottes, P.~Collet, and B.~Schmitt.
\newblock Statistical consequences of the {D}evroye inequality for processes.
  {A}pplications to a class of non-uniformly hyperbolic dynamical systems.
\newblock {\em Nonlinearity}, 18(5):2341--2364, 2005.

\bibitem{ChGo}
J.-R. Chazottes and S.~Gou{\"e}zel.
\newblock Optimal concentration inequalities for dynamical systems.
\newblock {\em To appear in Commun. Math. Phys}, 2012.

\bibitem{CL}
P.~Collet and Y.~Levy.
\newblock Ergodic properties of the {L}ozi mappings.
\newblock {\em Comm. Math. Phys.}, 93(4):461--481, 1984.

\bibitem{KS}
Holger Kantz and Thomas Schreiber.
\newblock {\em Nonlinear time series analysis}.
\newblock Cambridge University Press, Cambridge, second edition, 2004.

\bibitem{LN}
Steven~P. Lalley and A.~B. Nobel.
\newblock Denoising deterministic time series.
\newblock {\em Dyn. Partial Differ. Equ.}, 3(4):259--279, 2006.

\bibitem{Led}
Michel Ledoux.
\newblock {\em The concentration of measure phenomenon}, volume~89 of {\em
  Mathematical Surveys and Monographs}.
\newblock American Mathematical Society, Providence, RI, 2001.

\bibitem{Ser}
Regis~J. Serinko.
\newblock Ergodic theorems arising in correlation dimension estimation.
\newblock {\em J. Statist. Phys.}, 85(1-2):25--40, 1996.

\bibitem{Y1}
Lai-Sang Young.
\newblock Statistical properties of dynamical systems with some hyperbolicity.
\newblock {\em Ann. of Math. (2)}, 147(3):585--650, 1998.

\bibitem{Y2}
Lai-Sang Young.
\newblock What are {SRB} measures, and which dynamical systems have them?
\newblock {\em J. Statist. Phys.}, 108(5-6):733--754, 2002.
\newblock Dedicated to David Ruelle and Yasha Sinai on the occasion of their
  65th birthdays.

\end{thebibliography}

\end{document}